\documentclass[letterpaper,11pt]{article}
\usepackage[margin=1in]{geometry}
\usepackage[bookmarks, colorlinks=true, plainpages = false, colorlinks=true,
            citecolor=red,
            linkcolor=blue,
            anchorcolor=red,
            urlcolor=blue]{hyperref}
\usepackage{url}
\usepackage{amsmath,amsfonts,amsthm,amssymb,bm, verbatim,dsfont,mathtools}
\usepackage{color,graphicx,appendix}
\usepackage{etoolbox}

\usepackage{array}
\usepackage{multirow}

\usepackage{float}
\usepackage{tikz}
\usetikzlibrary{matrix,arrows,calc,shapes,backgrounds}
\usetikzlibrary{shapes.callouts,decorations.text}

\usetikzlibrary{shapes.misc}

\tikzset{cross/.style={cross out, draw=black, minimum size=2*(#1-\pgflinewidth), inner sep=0pt, outer sep=0pt},
cross/.default={1pt}}

\tikzstyle{int}=[draw, fill=blue!20, minimum size=2em]
\tikzstyle{dot}=[circle, draw, fill=blue!20, minimum size=2em]
\tikzstyle{dotred}=[circle, draw, fill=red!20, minimum size=2em]
\tikzstyle{init} = [pin edge={to-,thin,black}]
\tikzstyle{initred} = [pin edge={to-,thin,red}]
\tikzstyle{plan}=[draw, fill=blue!20, minimum size=2em, text width=5em, rounded corners,align=center]
\tikzstyle{planwide}=[draw, fill=blue!20, minimum size=2em, text width=8em, rounded corners,align=center]
\tikzstyle{nodedot}=[circle, draw, fill=white, minimum size=0.3cm,inner sep=0pt]
\tikzstyle{nodedot}=[circle, draw, fill=white, minimum size=3,inner sep=0pt]
\tikzstyle{Medge}=[green!60!black, thick]
\tikzstyle{Bedge}=[red, thick]
\tikzstyle{Cedge}=[blue, thick]
\tikzstyle{Sedge}=[black, thick]
\tikzstyle{Mgiantedge}=[green!60!black, line width=3.0pt]
\tikzstyle{Bgiantedge}=[red,line width=3.0pt]
\tikzstyle{Cgiantedge}=[blue,line width=3.0pt]
\tikzstyle{Sgiantedge}=[black,line width=3.0pt]
\tikzstyle{shadedgiantnode}=[circle, draw, fill=black!10, minimum size=1cm, inner sep=0pt]
\tikzstyle{unshadedgiantnode}=[circle, draw, fill=white, minimum size=1cm, inner sep=0pt]
\makeatletter
\tikzset{my loop/.style =  {to path={
  \pgfextra{}
  [looseness=5,min distance=10mm]
  \tikz@to@curve@path},font=\sffamily\small
  }}
\makeatletter
\newcolumntype{C}[1]{>{\centering\arraybackslash}p{#1}}

\tikzstyle{vertexdot}=[circle, draw, fill=white, minimum size=3,inner sep=0pt]
\tikzstyle{root}=[circle, draw, fill=black, minimum size=3,inner sep=0pt]

\tikzstyle{vertexdotsolid}=[circle, draw, fill=black, minimum size=3,inner sep=0pt]

\usepackage{pgfplots}

\pgfplotsset{
    standard/.style={
        axis x line=middle,
        axis y line=middle,
        every axis x label/.style={at={(current axis.right of origin)},anchor=north west},
        every axis y label/.style={at={(current axis.above origin)},anchor=north west}
    }
}
\usepgfplotslibrary{fillbetween}

\usepackage{xr,xspace}
\usepackage{todonotes}
\usepackage{paralist}
\usepackage{enumitem}

\usepackage{caption,subcaption,soul}
\usepackage[ruled]{algorithm2e}
\makeatletter
\usepackage{romanbar}
\usepackage{booktabs}
\usepackage{cleveref}
\usepackage{array}

\theoremstyle{plain}
\newtheorem{theorem}{Theorem}
\newtheorem{lemma}{Lemma}
\newtheorem{proposition}{Proposition}

\theoremstyle{definition}

\newtheorem{problem}{Problem}

\newtheorem{claim}{Claim}
\newtheorem*{remark*}{Remark}
\newtheorem*{theorem*}{Theorem}
\newtheorem{example}{Example}

\usepackage{xspace,prettyref}

\newcommand{\ones}{\mathbf 1}
\newcommand{\zeros}{\mathbf 0}
\newcommand{\reals}{{\mathbb{R}}}
\newcommand{\integers}{{\mathbb{Z}}}

\newcommand{\Expect}{\mathbb{E}}

\newcommand{\Prob}{\mathbb{P}}

\newcommand{\toprob}{\xrightarrow{\Prob}}

\def\Var{\mathrm{Var}}

\newrefformat{eq}{(\ref{#1})}
\newrefformat{chap}{Chapter~\ref{#1}}
\newrefformat{sec}{Section~\ref{#1}}
\newrefformat{alg}{Algorithm~\ref{#1}}
\newrefformat{fig}{Figure~\ref{#1}}
\newrefformat{tab}{Table~\ref{#1}}
\newrefformat{rmk}{Remark~\ref{#1}}
\newrefformat{clm}{Claim~\ref{#1}}
\newrefformat{claim}{Claim~\ref{#1}}
\newrefformat{def}{Definition~\ref{#1}}
\newrefformat{cor}{Corollary~\ref{#1}}
\newrefformat{lmm}{Lemma~\ref{#1}}
\newrefformat{prop}{Proposition~\ref{#1}}
\newrefformat{prob}{Problem~\ref{#1}}
\newrefformat{app}{Appendix~\ref{#1}}
\newrefformat{hyp}{Hypothesis~\ref{#1}}
\newrefformat{thm}{Theorem~\ref{#1}}
\newrefformat{fac}{Fact~\ref{#1}}
\newrefformat{ex}{Example~\ref{#1}}

\newcommand{\norm}[1]{\left\|{#1} \right\|}

\newcommand{\indc}[1]{{\mathbf{1}_{\left\{{#1}\right\}}}}

\newcommand{\diag}[1]{\mathsf{diag} \left\{ {#1} \right\} }

\newcommand{\calL}{{\mathcal{L}}}

\renewcommand{\bar}{\overline}

\newcommand{\bdx}{\boldsymbol{x}}

\newcommand{\bds}{\boldsymbol{s}}
\newcommand{\bdy}{\boldsymbol{y}}
\newcommand{\bdp}{\boldsymbol{p}}

\newcommand{\bdlambda}{\boldsymbol{\lambda}}

\usepackage{color-edits}
\addauthor[Mingwei]{mw}{purple}
\addauthor[Flore]{fs}{blue}

\newcommand{\weight}{W}

\title{When do token systems stabilize? Compatibility constraints and community structure}

\author{Flore Sentenac, Mingwei Yang, and Sophie H.\ Yu\thanks{F. Sentenac is with the Department of Information Systems and Operations Management, HEC Paris Business School, France, \texttt{sentenac@hec.fr}.
M. Yang is with the Department of Management Science and Engineering, Stanford University, Stanford CA, USA,
\texttt{mwyang@stanford.edu}.
S.\ H.\ Yu is with The Wharton School of Business, University of Pennsylvania, Philadelphia PA, USA,  \texttt{hysophie@wharton.upenn.edu}.
}}

\date{\today}

\begin{document}
\maketitle

\begin{abstract}
    Token systems support repeated service exchange without monetary transfers, but their viability depends on the stability of token balances.
    We study token systems with compatibility constraints among agents.
    A requester is randomly sampled in each period, and only a few randomly sampled agents are available to provide service.
    We establish a necessary and sufficient condition for stability and show that the minimum-token rule, which selects the available provider with the lowest token balance, achieves the optimal stability region.
    We then apply this stability criterion to a random token system generated from stochastic block models with uniformly sampled requesters.
    We tightly characterize the stability regime by identifying two obstructions to stability: insufficient service opportunities for a group of communities or for an individual agent.
\end{abstract}

\tableofcontents

\section{Introduction}
\label{sec:introduction}

Token systems provide a way to organize repeated service exchange when monetary transfers are unavailable or undesirable.
They have been implemented in babysitting cooperatives~\cite{sweeney1977monetary}, peer-to-peer systems~\cite{vishnumurthy2003karma}, and food banks~\cite{prendergast2022allocation}, among others.
In a typical token system, an agent earns a token by providing a service and spends a token when receiving one.
By recording past contributions, token systems provide a basis for sustaining cooperation and discouraging free riding when immediate or bilateral exchange is impractical.
A key concern, however, is the stability of token balances: agents who accumulate large token balances may have little incentive to continue providing service.

One application is cooperation between hospitals in kidney exchange.
Pooling patient--donor pairs across hospitals creates opportunities for exchanges that individual hospitals could not arrange on their own.
However, free riding has been documented in this setting: hospitals often match patient--donor pairs internally and submit primarily their harder-to-match pairs to centralized exchange platforms~\cite{ashlagi2014free,agarwal2019market,ashlagi2021kidney}.
This selective participation motivates the use of token systems that account for hospitals' contributions when assigning priority for future matching opportunities~\cite{ashlagi2026power}; see~\cite{hajaj2015strategy} for an earlier credit mechanism for multi-hospital kidney exchange.

Motivated by this setting, \cite{ashlagi2026power} study a stochastic model of repeated service exchange among a finite set of agents.
All agents start with a token balance of zero, and balances are allowed to become negative.
In each discrete period, a randomly chosen agent requests service, and $d$ potential providers are sampled independently with replacement.
These candidates represent the limited and fluctuating service opportunities available for a given request.
Among the sampled candidates, the agent with the lowest token balance provides the service and receives one token from the requester.

In the kidney-exchange interpretation, the requester is the hospital bringing an easy-to-match pair, while the potential providers are hospitals with compatible hard-to-match pairs.
An exchange decreases the token balance of the hospital contributing the easy-to-match pair and increases that of the hospital whose hard-to-match pair is matched.
Under this convention, a lower token balance corresponds to a larger net contribution to the shared pool.
By giving hospitals with lower balances priority for matches involving their hard-to-match pairs, the minimum-token rule rewards past contributions and can encourage hospitals to contribute easy-to-match pairs to the shared pool.

A central question for the viability of the system is whether token balances remain stable.
Formally, stability means that the initial all-zero state is positive recurrent: the expected time of the first return to this state after time zero is finite.
When both the requester and provider sampling distributions are uniform, \cite{ashlagi2026power} show that $d=1$ leads to instability, whereas $d=2$ already suffices for stability under the minimum-token rule.
Thus, a choice between just two potential providers can suffice to stabilize the system.
However, their model does not account for community structures in hospitals' exchange relationships.
Hospitals may, for example, restrict their partnerships to institutions with which they have previously cooperated, or institutions that are geographically close-by.
In this paper, we study how such community structures affect the stability of token systems.

\paragraph{Our contributions.}
We generalize the model of \cite{ashlagi2026power} by sampling the $d$ potential providers independently with replacement from a requester-specific distribution.
This notably captures compatibility constraints under which each agent can serve only a restricted subset of agents.
We establish a necessary and sufficient condition for the existence of an admissible stationary provider-selection rule that stabilizes the system (\Cref{thm:min-token-backpressure}).
Informally, the condition compares each subset's rate of requesting service with the probability that at least one of its members is sampled as a potential provider.
We also show that the minimum-token rule, described in \cite{ashlagi2026power} as natural and motivated by incentive considerations, is optimal for stability: whenever an admissible stationary provider-selection rule stabilizes the system, the minimum-token rule does as well.
In concurrent work, \cite{kerimov2026stability} obtains the corresponding characterization and optimality result when all requesters use the same strictly positive provider-sampling distribution.

To understand how community structure affects stability, we specialize our general stability criterion to stochastic block models~\cite{holland1983stochastic,DBLP:conf/focs/AbbeS15,abbe2018community}, in which agents are partitioned into communities and compatibility links are generated with probabilities determined by community membership.
We consider uniformly sampled requesters with $d\geq2$.
We first identify two obstructions to stability: a \emph{community obstruction}, when a group of communities has insufficient opportunities to provide service (\Cref{lem:block-cut-instability}), and an \emph{individual obstruction}, when a single agent has insufficient opportunities (\Cref{thm:instability-sbm}).
We then show that excluding both obstructions guarantees stability with high probability (\Cref{thm:stability-sbm}).
Our stability characterization covers all common degree scales: the individual obstruction is always present at sublogarithmic degree scales and, when the community condition holds, becomes vacuous at superlogarithmic degree scales, with the transition occurring at the logarithmic scale.

\subsection{Related literature}
\label{sec:related-work}

\paragraph{Token systems and service exchange.}
Many related models on token systems study service exchange with a fixed supply of tokens,
often called scrip, and nonnegative balances.
Requests may go unserved when the requester lacks the tokens needed
to pay for service~\cite{friedman2006efficiency,kash2012optimizing,
kash2015equilibrium,johnson2014analyzing}.
When agents maximize expected discounted benefits from receiving
services net of the costs of providing them,
\cite{friedman2006efficiency,kash2015equilibrium} establish
$\varepsilon$-Nash equilibria in sufficiently large populations
of sufficiently patient agents, in which agents volunteer only
when their token holdings fall below a threshold.
\cite{johnson2014analyzing} show that the minimum-scrip rule maximizes
welfare among provider-selection rules in their homogeneous model.
\cite{kanoria2024blind} study centralized control in such scrip systems.
Their policy selects a compatible provider with the lowest balance
and can reject requests to manage token circulation.
It achieves asymptotically optimal welfare without knowing demand
arrival rates.

In a slightly different vein, \cite{hu2026token} study favor trading
with dynamic entry and exit, and design token issuance, redemption,
and recycling mechanisms that can support an always-trade equilibrium.
Their mechanism bounds token holdings and uses monetary transfers
at balance boundaries and exit.

Our model builds on~\cite{ashlagi2026power}, where balances may be
negative and have no imposed bounds.
They show that the system is unstable with $d=1$, whereas the
minimum-token rule ensures stability with $d\geq2$ when requesters
and candidate providers are sampled uniformly.
They also give necessary and sufficient conditions for
stability with two agents who may differ in their request and
availability rates, assuming that candidate providers are sampled
independently of the requester's identity.
In concurrent work, \cite{kerimov2026stability} characterizes
the stability region for arbitrary strictly positive requester and
provider distributions under this independence assumption, and also
studies random service costs and stationary tail bounds.
Our focus is on requester-specific compatibility constraints
and stability in stochastic block models.

\paragraph{Kidney exchange.}
Credit mechanisms link successive matching rounds to encourage
hospitals to disclose their patient--donor
pairs~\cite{hajaj2015strategy} or to compensate countries whose past
transplant allocations fell below a prescribed fair
share~\cite{klimentova2021fairness,biro2020compensation}.
The opportunities to provide such compensation are constrained by
compatibility: hospitals and countries contribute pools that differ
in size and composition~\cite{ashlagi2026power,biro2020compensation}.
Random-graph models capture this heterogeneity through directed
compatibility edges between patient--donor pairs with type-dependent
probabilities~\cite{ashlagi2014free,ashlagi2019matching}.
These features motivate our study of whether token systems can
maintain balanced exchanges under heterogeneous compatibility
constraints.

\paragraph{Other types of nonmonetary exchange and credit.}
Besides token systems, cooperation and allocation without money are studied through repeated favor exchange~\cite{mobius2001trading,hauser2008trading,abdulkadiroglu2013trust}, artificial currencies~\cite{guo2009competitive,gorokh2021monetary}, karma mechanisms~\cite{censi2019today,elokda2024karma}, and promises of future allocations~\cite{balseiro2019multiagent}.
In credit networks, payments are routed along trust-graph paths with enough residual credit~\cite{dandekar2011liquidity}, possibly under per-agent borrowing caps~\cite{ramseyer2020liquidity}.
Our graph instead restricts who can serve whom, and we study positive recurrence under nonstrategic selection rules rather than incentives, efficiency, or liquidity.

\paragraph{Load balancing and network constraints.}
The minimum-token rule is the greedy rule of the power of $d$ choices~\cite{mitzenmacher2001survey}, which sharply reduces delays in queueing~\cite{vvedenskaya1996queueing,mitzenmacher2001power} and maximum loads in balls-into-bins~\cite{azar1999balanced}.
In the classical balls-into-bins model, the greedy rule is optimal among online $d$-choice strategies, as the minimum-token rule is for stability.
When candidate choices are constrained by a network, graphical variants draw candidates from a random edge~\cite{kenthapadi2006balanced,peres2015graphical,bansal2022power} or a hypergraph-structured random set~\cite{godfrey2008balls}.
Compatibility-constrained queueing is studied through large-system limits~\cite{mukherjee2018asymptotically,budhiraja2019supermarket,weng2020optimal,rutten2023load}.
In our model, the paying requester also determines the candidate distribution.

Graphical allocation has also been studied with signed, conserved loads: in a synchronous model with uniform-vertex removals, \cite{oleskertaylor2026graphical} bound equilibrium load differences on connected regular graphs.
In their asynchronous model, each arrival joins the (noisily) shortest nonnegative queue on its hyperedge, and an equilibrium exists if every vertex set serves faster than the hyperedges inside it receive arrivals.
The strict form of our cut condition~\eqref{eqn:cut-con-stability} gives the same subset-capacity inequality on complementary sets.

\paragraph{Queueing stability and scheduling.}
\cite{tassiulas1992stability} give throughput-optimal scheduling policies for constrained queueing networks, and \cite{tassiulas1993dynamic} show that serving the longest connected queue stabilizes parallel queues with random connectivity whenever any policy can.
A complementary question is how to characterize the stability region through subset-capacity conditions.
\cite{foley2001join} show that Markovian join-the-shortest-queue routing within server subsets is positive recurrent if and only if every subset has more capacity than the demand confined to it.
\cite{cruise2020stability} recover this characterization via Lyapunov drift.
Conditions of Hall type, similar to our cut condition, also arise in stochastic matching~\cite{caldentey2009fcfs}.
Match-the-longest is maximally stable in bipartite~\cite{busic2013stability} and non-bipartite~\cite{mairesse2016stability} matching models, as are generalized max-weight policies in the latter setting~\cite{jonckheere2023generalized}.

\paragraph{Stochastic block models.}
Stochastic block models, rooted in social network analysis~\cite{holland1983stochastic,wang1987stochastic}, provide a simple framework for networks with community structure.
Much of the literature~\cite{newman2018networks,abbe2018community} concerns recovering the blocks, for example by spectral clustering~\cite{lei2015consistency}.
We instead use the model to generate compatibility constraints, as do~\cite{sopranoloto2025online} for online matching in a dense regime.

\section{Generalized token systems}

We consider a general version of the model proposed by \cite{ashlagi2026power}.
There are $n$ agents, each with an integer token balance that may be positive or negative.
Write $s_{t,i}$ for the token balance of agent $i$ at the end of time $t$, and define $\bds_t:=(s_{t,1},\ldots,s_{t,n})^{\mathsf T}\in\integers^n$ as the vector of all agents' token balances.
At time $0$, each agent has $0$ tokens, i.e., $\bds_0 = \zeros$.
At each time $t \in \integers_{> 0}$, a service requester $r_t$ is sampled from a distribution $\bdp \in \reals_{> 0}^n$.
Conditional on $r_t$, we draw $d$ candidate service providers independently with replacement to form a candidate pool, which may contain repeated agents.

Specifically, let $\weight \in \reals_{\geq0}^{n \times n}$ be a weighted directed adjacency matrix of provider availability, with $\weight(i,j)$ denoting the weight of provider $j$ for requester $i$.
For $S\subseteq[n]$, write
\[
    \weight(i,S):=\sum_{j\in S}\weight(i,j),\qquad \weight(i,[n])>0\quad\text{for every }i.
\]
Requester $i$ samples providers identically and independently using probabilities $\weight(i,j)/\weight(i,[n])$.
In this model, we allow some agents to serve themselves, i.e. $W(i,i)>0$ is permitted.
Denote $A_t$ as the drawn set of potential providers at time $t$.

A provider-selection rule chooses $a_t\in A_t$ according to a time-homogeneous distribution that may depend only on $r_t$, $A_t$, and $\bds_{t-1}$.
The requester gives a token to the provider.
In other words, if $a_t \neq r_t$, then $s_{t, r_t} = s_{t - 1, r_t} - 1$ and $s_{t, a_t} = s_{t - 1, a_t} + 1$ with $s_{t, i} = s_{t - 1, i}$ for all $i \in [n] \setminus \{a_t, r_t\}$; otherwise, when $a_t = r_t$, we have $\bds_t = \bds_{t - 1}$.

Under an admissible stationary policy, we say that the token system is \emph{stable} if its initial zero-balance state is positive recurrent, that is, $\Expect_{\zeros}[\tau_{\zeros}^{+}]<\infty$, where $\tau_{\zeros}^{+}:=\inf\{t\geq1:\bds_t=\zeros\}$.
This implies that the communicating class containing $\zeros$ is closed and admits a stationary probability distribution.
Stability thus gives finite expected return to the initial balanced state and a stationary regime for token imbalances.

\subsection{Stability criterion}
\label{sec:stable-fix}

We first give a deterministic criterion to certify that a token system can be stabilized.
The remaining modeling choice is the provider-selection rule.
The minimum-token rule selects a candidate with the fewest tokens.
This is the rule studied by \cite{ashlagi2026power}, who break ties uniformly at random.
We show that the minimum-token rule stabilizes every system that can be stabilized by an admissible provider-selection rule.
The result holds for any admissible stationary tie-breaking rule.

Before stating the theorem, we introduce the following notion: for $\emptyset \neq S \subsetneq [n]$, we say that $S$ is \emph{sampling-separating} if
\begin{align*}
    W(i, S) =
    \begin{cases}
        W(i, [n]), & i \in S, \\
        0, & i \notin S.
    \end{cases}
\end{align*}
In other words, the agents in $S$ can only be served by themselves, and they cannot serve agents outside $S$.
The proof of \Cref{thm:min-token-backpressure} is deferred to \Cref{sec:proof-thm-min-token-backpressure}.

\begin{theorem}\label{thm:min-token-backpressure}
    If, for every $\emptyset\neq S\subsetneq[n]$,
    \begin{align}
        \Pr[A_t\cap S\neq\emptyset]
        \geq\sum_{i\in S}p_i,
        \label{eqn:cut-con-stability}
    \end{align}
    and each set attaining equality is sampling-separating, then the minimum-token rule stabilizes the system.
    Otherwise, no admissible stationary policy stabilizes the system.
\end{theorem}

The proof of the stability of the minimum token rule  proceeds by showing that the quadratic Lyapunov function admits a negative drift under the minimum-token rule, and this argument is closely related to the analysis of the backpressure algorithm for multi-hop queueing networks~\cite{tassiulas1992stability,georgiadis2006resource}.
In particular, both the minimum-token rule and the backpressure algorithm can be naturally derived by greedily minimizing the drift of the quadratic Lyapunov function.
This negative drift is shown by establishing that, for each subset $S$, if \eqref{eqn:cut-con-stability} holds strictly, whenever all agents in $S$ own too few tokens, their overall probability of providing services exceeds their overall probability of requesting services, so that the agents in $S$ will be gaining tokens on average.
We remark that a closely related summation-by-parts decomposition and a quadratic-drift argument also appear in the concurrent work of \cite[Proposition~1 and the proof of Theorem~1]{kerimov2026stability}.

On the other hand, if \eqref{eqn:cut-con-stability} fails, the total number of tokens owned by agents in $S$ has a uniformly negative drift.
If equality holds and $S$ is not sampling-separating, that total is bounded above pathwise by a non-degenerate balanced random walk.
Both cases rule out positive recurrence, as shown in \Cref{sec:proof-thm-min-token-backpressure}; equality at a sampling-separating cut is harmless because its token total is conserved.

The following example shows that even a small change in the compatibility graph can alter stability.

\begin{example}\label{exa:unstable-k14}
    Consider five agents arranged as the star graph $K_{1, 4}$, where agent $1$ is the center, and the remaining agents are leaves.
    Suppose that the requester is sampled uniformly at random, and $d = 2$.
    Each available service provider is sampled uniformly at random from the closed neighborhood of the requester.
    We show that the system is unstable by applying \Cref{thm:min-token-backpressure}, as witnessed by the cut $S = \{2, 3, 4, 5\}$:
    \begin{align*}
        \Pr[A_t \cap S \neq \emptyset]
        &= \frac{1}{5} \Pr[A_t \cap S \neq \emptyset \mid r_t = 1] + \frac{4}{5} \Pr[A_t \cap S \neq \emptyset \mid r_t \in S] \\
        &= \frac{1}{5} \cdot \frac{24}{25} + \frac{4}{5} \cdot \frac{3}{4}
        = \frac{99}{125}
        < \frac{100}{125}
        = \sum_{i \in S} p_i.
    \end{align*}
    On the other hand, one can easily verify by applying \Cref{thm:min-token-backpressure} that the system consisting of four agents arranged as the star graph $K_{1, 3}$ is stable.
\end{example}

\section{Token systems with community structures}
\label{sec:sbm-model}

Motivated by token systems where service opportunities can depend on social or organizational groups, a natural next step is to look at how common graph structures fare in terms of stability.
For example, babysitting cooperatives may coordinate requests through separate sections~\cite{sweeney1977monetary}, while social clouds permit sharing within friendship networks or specified groups~\cite{chard2010social}.
In addition, block structures have also been identified in time-bank transactions using stochastic block models~\cite{gluckler2022timebanks}.
We use a symmetric stochastic block model~\cite{holland1983stochastic,DBLP:conf/focs/AbbeS15,abbe2018community} to study how group-dependent service opportunities affect token-balance stability.

After introducing the model in \Cref{sec:sbm-setup}, we show that stability with high probability is governed by two obstructions.
We first describe the \emph{community obstruction} in \Cref{sec:sbm-community}: a group of whole communities requests service more often than at least one of its members appears among the available providers.
We then describe the \emph{individual obstruction} in \Cref{sec:sbm-individual-obstruction}, where a single agent appears among the available providers less often than she requests service.
To clearly convey the intuition behind individual obstructions, we start with Erd\H{o}s--R\'enyi graphs in \Cref{sec:sbm-er}, and then turn to the general stochastic block models in \Cref{sec:sbm-individual}.
Finally, we show in \Cref{sec:sbm-main} that the absence of both obstructions guarantees stability with high probability.

\subsection{Model and stability condition}
\label{sec:sbm-setup}

Consider a community structure with $m$ disjoint blocks $H_1,\ldots,H_m$, denote $n_u = |H_u|$ and $\frac{n_u}{n}=\alpha_u$ for each $u \in [m]$, and suppose that for all agents $i\neq j$ with $i \in H_u$ and $j \in H_v$,
\[
\mathbb{P}\left(W(i,j)=1\right)= \mathbb{P}\left(W(j,i)=1\right)= \frac{c_{uv} L_n}{n} =: Q_{uv}
\]
for some constant $c_{uv}=c_{vu}>1$.
Edges for distinct unordered pairs are sampled independently.
Throughout this section, $m\geq1$, $d\geq2$, the positive proportions $\alpha_u$, and the coefficients $c_{uv}$ are fixed, with $n_u=\alpha_un$ along admissible integers.
We allow any deterministic sequence $L_n\geq L_0>0$ such that $c_{\max}L_n/n\leq1$, where $c_{\max}:=\max_{u,v}c_{uv}$.
Also suppose that for each $i \in [n]$,
\[
p_i=\frac{1}{n} \text{ and }W(i,i)=1.
\]
Self-service, i.e. $W(i,i)=1$,  is imposed to ensure that the system is well-defined even with isolated vertices.
Although we do not pursue this extension here, similar results could be derived without this technical assumption, since graphs in the stability regime of the model above have no isolated vertices with high probability.

Let $G$ denote the resulting undirected graph on $[n]$ (with the self-loops removed), let $N(u)$ be the set of neighbors of agent $u$ in $G$, $N(u,H_j):=N(u)\cap H_j$, $\bar N(u):=N(u)\cup\{u\}$, and let $a_u:=W(u,[n])=|\bar N(u)|$ be the closed degree of $u$.

For $x\in[0,1]$, define $h_d(x):=1-(1-x)^d$.
This is the probability that at least one of $d$ independent draws hits a set having probability $x$ under a single draw.
In particular, for a fixed graph realization,
\[
 \Pr(A_t\cap S\neq\emptyset\mid r_t=u)
 =h_d\!\left(\frac{W(u,S)}{a_u}\right).
\]

Since requests are uniformly distributed, \Cref{thm:min-token-backpressure} implies stability under the minimum-token rule if
\begin{align}
    \frac{1}{n} \sum_{u=1}^n h_d\left(\frac{W(u,S)}{a_u}\right)> \frac{|S|}{n},
    \qquad\forall\,\emptyset\neq S\subsetneq[n];
    \tag{Cut-Uniform}
    \label{eq:er-closed-neighborhood-cut-SBM}
\end{align}
in addition, a strict reverse inequality for even one cut implies instability under every admissible provider-selection rule.
Although this criterion is difficult to interpret directly, we show that it reduces to two conditions with clear interpretations.

\subsection{Community obstruction}
\label{sec:sbm-community}

Consider a cut consisting of several whole communities.
If these communities request service more often than any of their members appears among the candidates, no provider-selection rule can prevent their aggregate token amount from drifting downward.
This is the stability requirement at the global scale, when the deficient set consists of whole communities.

To formally express this obstruction, define
\[
    R(i,j):=\frac{\alpha_jc_{ij}}{\sum_{k=1}^m\alpha_kc_{ik}},
\]
which, when $L_n\to\infty$, is the asymptotic probability that a candidate drawn by a requester in community $i$ lies in community $j$.
Note that the matrix $R$ is row-stochastic and has positive entries.
Let $\bdx\in\{0,1\}^m$ indicate the communities in a cut $S=\bigcup_{i:x_i=1}H_i$.
Then, $(R\bdx)_i$ is the asymptotic probability that one candidate drawn by a requester in community $i$ lies in $S$.
When $L_n\to\infty$, the LHS of Condition~\eqref{eq:er-closed-neighborhood-cut-SBM} for cut $S$ can be written as
\[
    \frac1n\sum_{u=1}^n h_d\left(\frac{W(u,S)}{a_u}\right)
    \approx \frac{1}{n} \sum_{i=1}^m n_i h_d((R \bdx)_i)
    =\sum_{i=1}^m\alpha_i h_d((R\bdx)_i),
\]
while the RHS is $|S|/n=\sum_i\alpha_i x_i$.
Therefore, \eqref{eq:er-closed-neighborhood-cut-SBM} for all cuts consisting of whole communities gives
\begin{align}
    \sum_{i=1}^m\alpha_i h_d((R\bdx)_i)
    >\sum_{i=1}^m\alpha_i x_i,
    \qquad\forall\,\bdx\in\{0,1\}^m\setminus\{\zeros,\ones\}.
    \tag{Block-Cut}
    \label{eqn:bc-general-sbm}
\end{align}

We show in the following proposition that a strict reverse inequality in \eqref{eqn:bc-general-sbm} implies instability with high probability under every admissible provider-selection rule.
As a result, \eqref{eqn:bc-general-sbm} is, up to the equality case, a necessary condition for stability.
The proof of \Cref{lem:block-cut-instability} is given in \Cref{sec:proof-lem-block-cut-instability}.

\begin{proposition}\label{lem:block-cut-instability}
    Suppose that $G$ is sampled from the stochastic block model and $d\geq2$.
    If, for some $\bdx\in\{0,1\}^m\setminus\{\zeros,\ones\}$,
    \begin{align}
        \sum_{i=1}^m\alpha_i h_d((R\bdx)_i)
        <\sum_{i=1}^m\alpha_i x_i,
        \label{eqn:bc-general-violate}
    \end{align}
    then the system is unstable with high probability under every admissible provider-selection rule.
\end{proposition}

\subsection{Individual obstruction}
\label{sec:sbm-individual-obstruction}

\subsubsection{Erd\H{o}s--R\'enyi graphs}
\label{sec:sbm-er}

To see that the block condition \eqref{eqn:bc-general-sbm} alone is not enough, suppose that $m=1$, $c_{11}=\lambda$, and $L_n=\ln n$.
That is exactly an Erd\H{o}s--R\'enyi graph with edge probability $\lambda \ln(n) / n$.
There is only one community, so \eqref{eqn:bc-general-sbm} is vacuous, yet, as we will see, stability still depends on $\lambda$.
In other words, \eqref{eqn:bc-general-sbm} does not prevent an individual agent from being sampled too rarely, revealing an obstruction at the local scale.

To formally present the obstruction, recall that proving stability entails examining the condition in \eqref{eq:er-closed-neighborhood-cut-SBM} for all cuts.
For Erd\H{o}s--R\'enyi graphs, we show that singleton cuts serve as the bottleneck when determining the degree threshold for stability.

By Taylor expansion, the LHS of Condition~\eqref{eq:er-closed-neighborhood-cut-SBM} for singleton cut $\{u\}$ can be approximated by
\begin{align*}
    \frac{1}{n} \sum_{v=1}^n h_d\left(\frac{W(v,u)}{a_v}\right)
    \approx \frac{d}{n} \sum_{v=1}^n \frac{W(v, u)}{a_v}.
\end{align*}
Define
\begin{align}
    B_u:=\sum_{v=1}^n\frac{W(v,u)}{a_v}
    =\sum_{v\in\bar N(u)}\frac1{a_v}
    \approx \frac{a_u}{\lambda \ln n},
    \label{eqn:def-column-exposure-er}
\end{align}
and note that $B_u/n$ equals the probability of $u$ appearing in one candidate draw, averaged over the uniform requester.
We call $B_u$ the \emph{column exposure} of agent $u$.

For a singleton cut $\{u\}$, \eqref{eq:er-closed-neighborhood-cut-SBM} requires
\[
 \frac{1}{n} \sum_{v=1}^n h_d\left(\frac{W(v,u)}{a_v}\right)
 \approx \frac{d}{n} B_u
 > \frac{1}{n},
\]
and hence $dB_u < 1$ should imply that \eqref{eq:er-closed-neighborhood-cut-SBM} cannot be satisfied.
To formalize this statement, the approximation in \eqref{eqn:def-column-exposure-er} suggests comparing the minimum normalized degree, defined as $\min_{u \in [n]} a_u / \lambda\ln n$, with $1/d$.
Standard minimum-degree estimates for Erd\H{o}s--R\'enyi graphs (see, e.g.,~\cite[Chapter~3]{bollobas2001random} and~\cite[Lemma~3.4]{kolokolnikov2014algebraic}) give
\[
    \min_{u \in [n]}\frac{a_u}{\lambda\ln n}\toprob\tau,
    \qquad I(\tau)=1,\qquad
    I(t):=\lambda(1-t+t\ln t),
\]
with $0\ln0=0$.
Since $I$ is strictly decreasing on $(0,1)$,
\begin{align*}
    \tau>\frac1d
    \quad\iff\quad I\left(\frac1d\right)>1
    \quad\iff\quad\lambda>\left(1-\frac{1+\ln d}{d}\right)^{-1},
\end{align*}
resulting in the degree threshold for stability.

We formalize the stability criterion for Erd\H{o}s--R\'enyi graphs in the following theorem.

\begin{theorem}\label{thm:special-equal-degree-coeffi}
    Fix $d\geq2$ and $\lambda>1$, and let $G$ be sampled from the stochastic block model with $m=1$, $c_{11}=\lambda$, and $L_n=\ln n$.
    Define
    \[
        \lambda_d^*:=\left(1-\frac{1+\ln d}{d}\right)^{-1}.
    \]
    If $\lambda>\lambda_d^*$, the minimum-token rule stabilizes the system with high probability.
    If $\lambda<\lambda_d^*$, the system is unstable with high probability under every admissible provider-selection rule.
\end{theorem}

\Cref{thm:special-equal-degree-coeffi}, whose proof is deferred to \Cref{sec:proof-cor-special-equal-degree-coeffi}, follows by specializing \Cref{thm:instability-sbm} and \Cref{thm:stability-sbm} to a single community.
\Cref{fig:phase-transition} illustrates the stability regime for Erd\H{o}s--R\'enyi graphs.

\begin{figure}[!t]
\centering
\begin{tikzpicture}[
    declare function={lamstar(\d)=1/(1-(1+ln(\d))/\d);}
]
\begin{axis}[
    width=0.60\textwidth,
    height=0.50\textwidth,
    xmin=1, xmax=7,
    ymin=2, ymax=12,
    axis lines=left,
    axis line style={black!75, semithick},
    xlabel={$\lambda=np/\ln n$},
    ylabel={$d$},
    ylabel style={rotate=-90},
    xtick={1,2,3,4,5,6,7},
    ytick={2,4,6,8,10,12},
    tick align=outside,
    tick style={black!65},
    tick label style={font=\small},
    label style={font=\normalsize},
    enlargelimits=false,
    clip=false,
    axis on top
]

\path[fill=red!20, draw=none]
    (axis cs:1,2) rectangle (axis cs:7,12);

\path[fill=green!20, draw=none]
    (axis cs:{lamstar(12)},12)
    -- (axis cs:7,12)
    -- (axis cs:7,2)
    -- (axis cs:{lamstar(2)},2)
    -- plot[domain=2:12, samples=300, variable=\t]
        (axis cs:{lamstar(\t)},\t)
    -- cycle;

\addplot[
    black!85, thick, dashed,
    domain=2:12, samples=300, variable=\t
] ({lamstar(\t)},\t);

\addplot[
    only marks, mark=*, mark size=2pt,
    mark options={
        solid, draw=black!85, fill=white, line width=0.8pt
    },
    samples at={2,3,4,5,6,7,8,9,10,11,12},
    variable=\t
] ({lamstar(\t)},\t);

\node[align=center, font=\footnotesize] at (axis cs:2.0,2.75)
    {\textbf{Unstable w.h.p.}\\[3pt]
     $\lambda<\lambda_d^*$};

\node[align=center, font=\footnotesize] at (axis cs:4.9,9.3)
    {\textbf{Stable w.h.p.}\\[3pt]
     $\lambda>\lambda_d^*$};

\node[anchor=west, font=\small, inner sep=2pt]
    at (axis cs:3.3,6.1)
    {$\displaystyle
      \lambda_d^*=\left(1-\frac{1+\ln d}{d}\right)^{-1}$};

\draw[->, black!65, thin]
    (axis cs:3.22,6.1)
    -- (axis cs:{lamstar(5.4)},5.4);

\end{axis}
\end{tikzpicture}

\caption{Stability regimes for $G\sim G(n,p)$, with $\lambda=np/\ln n$.
The two regions indicate stability and instability with high probability.
Open circles mark the thresholds for integer $d\ge2$; the dashed curve interpolates between them.
The critical equality is not covered.
As $d\to\infty$, $\lambda_d^*\to1$.}
\label{fig:phase-transition}
\end{figure}
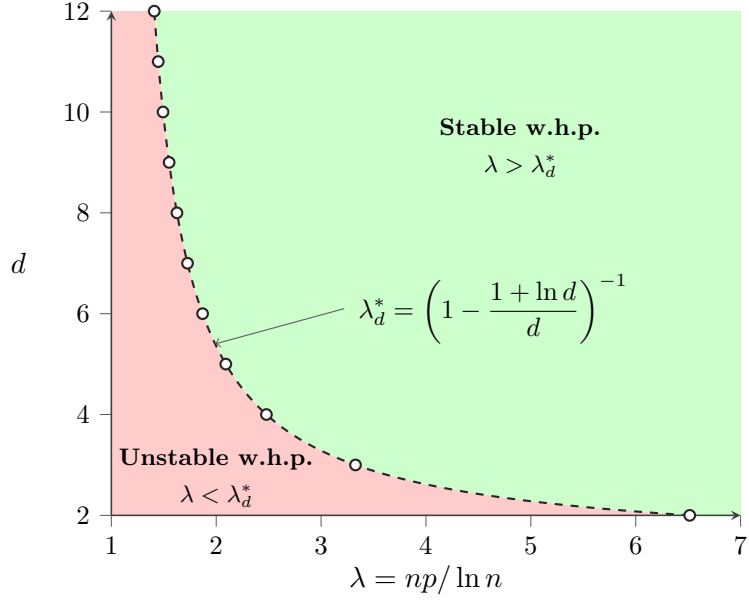

\subsubsection{Stochastic block models}
\label{sec:sbm-individual}

For stochastic block models, the approximation in \eqref{eqn:def-column-exposure-er} needs to be properly generalized since neighbors from different communities admit different expected degrees.
To account for these heterogeneous contributions from different neighbors to the exposure, define the expected closed-degree coefficient of community $i$ by
\[
    \lambda_i:=\sum_{j=1}^m\alpha_jc_{ij}>1.
\]
For intuition, assume $L_n \to \infty$ and $L_n=o(n)$.
Each $v \in H_i$ satisfies $\Expect a_v \approx \lambda_i L_n$, and hence the approximation in \eqref{eqn:def-column-exposure-er} becomes
\[
 B_u\approx\frac1{L_n}\sum_{j=1}^m\frac{W(u,H_j)}{\lambda_j}.
\]
For $u\in H_i$, a large-deviation heuristic for its community neighbor counts suggests
\begin{align*}
    \Prob(B_u \leq t) \approx \exp(- L_n I_i(t)),
    \qquad I_i(t):=\min_{\substack{\bdx\in\reals^m_{\geq0}\\\sum_jx_j/\lambda_j\leq t}}
    \sum_{j=1}^m\left(x_j\ln\frac{x_j}{\alpha_jc_{ij}}-x_j+\alpha_jc_{ij}\right).
\end{align*}
If, for some fixed $\epsilon>0$, there exists a community $i$ such that
\[
 \limsup_{n\to\infty}\frac{L_nI_i(1/d)}{\ln n} \leq 1 - \epsilon,
\]
then $\Prob(B_u \leq 1 / d) \approx \exp(-(1 - \epsilon)\ln n) = n^{-(1 - \epsilon)}$ for each $u \in H_i$.
We then show via a second-moment argument that, with high probability, there exists $u \in H_i$ with $d B_u < 1$.
Recall that $dB_u < 1$ witnesses the failure for the singleton $\{u\}$ to satisfy \eqref{eq:er-closed-neighborhood-cut-SBM}, concluding instability.
We formalize the above arguments for the individual obstruction in the following proposition, whose proof is given in \Cref{sec:proof-thm-instability-sbm}.

\begin{proposition}\label{thm:instability-sbm}
If there exists $i \in [m]$ such that
\[
 \limsup_{n\to\infty}\frac{L_nI_i(1/d)}{\ln n} <1,
\]
then the system is unstable with high probability under every admissible provider-selection rule.
\end{proposition}

\subsection{Combining the two obstructions}
\label{sec:sbm-main}

The following theorem shows that the absence of both obstructions with strict margins is sufficient for stability with high probability under the minimum-token rule.

\begin{theorem}\label{thm:stability-sbm}
If \eqref{eqn:bc-general-sbm} holds and
\[
    \liminf_{n\to\infty}\frac{L_n I_i(1/d)}{\ln n}>1
\]
for every $i \in [m]$, then the minimum-token rule stabilizes the system with high probability.
\end{theorem}

The proof of \Cref{thm:stability-sbm} is given in \Cref{sec:proof-thm-stability-sbm}.
Together with \Cref{lem:block-cut-instability,thm:instability-sbm}, these results give the stability criterion for all common degree scales.
In particular, \Cref{thm:instability-sbm} gives instability with high probability at sublogarithmic scales.
At superlogarithmic scales, \Cref{thm:stability-sbm} gives stability with high probability whenever \eqref{eqn:bc-general-sbm} holds.
At the logarithmic scale $L_n=\ln n$, $\min_{i\in[m]}I_i(1/d)=1$ serves as the critical threshold for individual obstruction.

\section{Proof of \Cref{thm:min-token-backpressure}}
\label{sec:proof-thm-min-token-backpressure}

For each $S\subseteq[n]$, define its candidate-hitting probability by
\begin{align*}
        f(S):=\Pr(A_t\cap S\neq\emptyset)
    =\sum_{i=1}^n p_i h_d \left(\frac{\weight(i,S)}{\weight(i,[n])}\right) \,, 
\end{align*}
and define the request probability by $p(S):=\sum_{i\in S}p_i$.

\paragraph{Stability under the minimum-token rule.}
Assume that every cut satisfies \eqref{eqn:cut-con-stability} and that each cut attaining equality is sampling-separating.
We show that the system is stable under the minimum-token rule for any admissible tie-breaking rule.

Let $\mathcal C$ be the connected components of the undirected support graph with an edge $\{i,j\}$ for distinct $i,j$ whenever $\weight(i,j)>0$ or $\weight(j,i)>0$.
Transfers stay within a component, so the initial condition implies
\[
    \sum_{i\in C}s_{t,i}=0\quad\text{for every }C\in\mathcal C,\ t\geq0.
\]
The Markov chain $(\bds_t)_{t \geq 0}$ therefore remains in the countable invariant state space
\[
    \mathcal X:=\left\{\bds\in\integers^n:       \sum_{i\in C}s_i=0\ \text{for every }C\in\mathcal C\right\}.
\]
If every support component is a singleton, the chain stays at $\zeros$ and is positive recurrent; henceforth, assume that some component contains at least two agents.

We first establish irreducibility on $\mathcal X$.
For any $C\in\mathcal C$ with $|C|\geq2$, each nonempty proper subset $S\subsetneq C$ is not sampling-separating, so the assumed cut criterion \eqref{eqn:cut-con-stability} gives a strict inequality.
Consider the directed support graph with an edge $i\to j$ when $\weight(i,j)>0$, and we claim that the subgraph induced by $C$ is strongly connected.
If the subgraph induced by $C$ were not strongly connected, let $S\subsetneq C$ be a source strongly connected component.
Then $A_t\cap S\neq\emptyset$ would imply $r_t\in S$, contradicting the strict cut inequality for $S$.
For every supported edge $i\to j$ with $i\neq j$, the event $(r_{t+1}=i) \land (A_{t+1}=\{j,\ldots,j\})$ has probability $p_i\bigl(\weight(i,j)/\weight(i,[n])\bigr)^d>0$ and forces the transition $\bds\to\bds-e_i+e_j$, independently of the state and tie-breaking rule.
Transfers along a directed path move one token between its endpoints, leaving the intermediate balances unchanged.
By concatenating such paths, any vector in a component's zero-sum integer lattice can be reached from any other.
Applying these transfers successively within the components proves irreducibility on $\mathcal X$.

Sampling-separating sets are precisely unions of components in $\mathcal C$.
Define the minimum slack over all non-separating cuts by
\begin{align*}
    \epsilon:=\min_{\substack{\emptyset\neq S\subsetneq[n]\\
    S\text{ is not sampling-separating}}}
    \bigl(f(S)-p(S)\bigr)>0.
\end{align*}
Consider the quadratic Lyapunov function
\begin{align*}
    \calL(\bds_t)
    := \frac{1}{2} \norm{\bds_t}_2^2,
\end{align*}
and we next establish a drift bound for $\calL(\bds_t)$.

For every transition,
\begin{align}
    \calL(\bds_{t + 1}) - \calL(\bds_t)
    &= \frac{1}{2} \langle \bds_{t + 1} + \bds_t, \bds_{t + 1} - \bds_t \rangle \notag \\
    &= \langle \bds_t, \bds_{t + 1} - \bds_t \rangle + \frac{1}{2} \norm{\bds_{t + 1} - \bds_t}_2^2
    \leq \langle \bds_t, \bds_{t + 1} - \bds_t \rangle + 1.
    \label{eqn:drift-quad-backpressure-det}
\end{align}
For fixed requester $r_{t+1}$ and candidate multiset $A_{t+1}$, the linear term in \eqref{eqn:drift-quad-backpressure-det} is
\begin{align}
    \langle \bds_t, \bds_{t + 1} - \bds_t \rangle
    = s_{t,a_{t+1}}-s_{t,r_{t+1}},
    \label{eqn:linear-term-drift}
\end{align}
and the minimum-token rule minimizes this term for every realization by choosing a candidate with the fewest tokens.
The strict cut inequalities then provide a strict negative drift needed for stability, as we formalize below.

Fix the current state $\bds_t=\bds\in\mathcal X$.
Choose a permutation $\pi$ of $[n]$ such that $s_{\pi(1)} \leq \ldots \leq s_{\pi(n)}$, and let $S_j := \{\pi(1), \ldots, \pi(j)\}$ consist of the first $j$ agents in $\pi$.
Whenever $s_{\pi(j)}<s_{\pi(j + 1)}$, we claim that the set $S_j$ is not sampling-separating.
To see this, if $S_j$ were sampling-separating, both $S_j$ and $[n]\setminus S_j$ would be unions of components in $\mathcal C$.
Since $\bds\in\mathcal X$ implies $\sum_{i\in C}s_i=0$ for each $C\in\mathcal C$, summing over these unions gives $\sum_{k=1}^j s_{\pi(k)}=0$ and $\sum_{k=j+1}^n s_{\pi(k)}=0$.
Consequently,
\[
    0=\frac1j\sum_{k=1}^j s_{\pi(k)}
    \leq s_{\pi(j)}<s_{\pi(j + 1)}
    \leq\frac1{n-j}\sum_{k=j+1}^n s_{\pi(k)}=0,
\]
leading to a contradiction.

For each $j \in [n - 1]$ such that $s_{\pi(j)} < s_{\pi(j + 1)}$, notice that the minimum-token rule selects a provider in $S_j$ precisely when some available provider belongs to $S_j$, and hence
\begin{align}
    \Pr[a_{t+1}\in S_j\mid\bds_t=\bds]&=\Pr[A_{t+1}\cap S_j\neq\emptyset].
    \label{eqn:prob-at+1-in-Sj}
\end{align}
For every agent $i\in[n]$, we have
\begin{align}
    s_i=s_{\pi(n)}-\sum_{j=1}^{n-1}(s_{\pi(j + 1)}-s_{\pi(j)})\indc{i\in S_j}.
    \label{eqn:si-sum-by-part}
\end{align}
Taking conditional expectations in \eqref{eqn:drift-quad-backpressure-det} and applying \eqref{eqn:linear-term-drift} yields
\begin{align*}
    \Expect[\calL(\bds_{t+1})-\calL(\bds_t)\mid\bds_t=\bds]
    &\leq 1+\Expect[s_{a_{t+1}}-s_{r_{t+1}}\mid\bds_t=\bds]\\
    &=1-\sum_{j=1}^{n-1}(s_{\pi(j + 1)}-s_{\pi(j)})
       \left(\Pr[a_{t+1}\in S_j\mid\bds_t=\bds]-\sum_{i\in S_j}p_i\right)\\
    &=1-\sum_{j=1}^{n-1}
       (s_{\pi(j + 1)}-s_{\pi(j)})\left(\Pr[A_{t+1}\cap S_j\neq\emptyset]-\sum_{i\in S_j}p_i\right)\\
    &\leq 1-\epsilon\sum_{j=1}^{n-1}(s_{\pi(j + 1)}-s_{\pi(j)}) \\
    &=1-\epsilon(s_{\pi(n)}-s_{\pi(1)}),
\end{align*}
where the first equality holds by applying \eqref{eqn:si-sum-by-part} to the requester and the provider, the second equality holds by \eqref{eqn:prob-at+1-in-Sj}, and the last inequality holds by the definition of $\epsilon$.
Since $\sum_{i=1}^n s_i=0$, we have $s_{\pi(1)}\leq0\leq s_{\pi(n)}$ and therefore
\[
    2\calL(\bds_t)=\sum_{i=1}^n s_{\pi(i)}^2\leq n(s_{\pi(n)}-s_{\pi(1)})^2.
\]
It follows that
\[
\Expect[\calL(\bds_{t+1})-\calL(\bds_t)\mid\bds_t]
\leq1-\epsilon\sqrt{\frac{2\calL(\bds_t)}{n}}.
\]
Taking $K:=2n/\epsilon^2$ gives drift at most $-1$ whenever $\calL(\bds_t)>K$.
The set $\{\bds\in\mathcal X:\calL(\bds)\leq K\}$ is finite because its coordinates are bounded integers.
From each state there are finitely many possible next states, each with finite quadratic value, so the required one-step expectation is finite.
The chain is time-homogeneous and irreducible on $\mathcal X$ as proved above.
Thus the Foster--Lyapunov criterion \cite[Theorem 3.1]{tassiulas1992stability} (see also~\cite[Chapter~11]{meyn2009markov} and~\cite[Chapter~8]{robert2013stochastic}) applies with drift bound $-1$, proving stability under the minimum-token rule with arbitrary permitted tie-breaking.

\paragraph{Instability when the cut criterion fails.}
Suppose there exists $\emptyset \neq S \subsetneq [n]$ such that either \eqref{eqn:cut-con-stability} does not hold, or \eqref{eqn:cut-con-stability} holds with equality and $S$ is not sampling-separating.

Consider any admissible stationary provider-selection rule, and we show that $\Expect_{\zeros}[\tau_{\zeros}^{+}]=\infty$ under this rule.
Let $\mathcal F_t$ be the sigma-field generated by the fixed matrix $\weight$ and the entire history through time $t$, including the selection randomness.
Let $X_t:=\sum_{i\in S}s_{t,i}$.
Conditional on the entire history, fresh sampling gives
\begin{align*}
    \Expect[X_{t+1}-X_t\mid\mathcal F_t]
    \leq f(S)-p(S)\leq0.
\end{align*}
Thus $(X_t)$ is a supermartingale with respect to $(\mathcal F_t)$, with $X_0=0$ and increments bounded in absolute value by one.

Let $\tau:=\tau_{\zeros}^+$, and suppose for contradiction that $\Expect_{\zeros}[\tau]<\infty$.
For each integer $t\geq1$, the stopping time $t\wedge\tau=\min\{t,\tau\}$ satisfies $1\leq t\wedge\tau\leq t$.
Since $(X_t)$ is a supermartingale, the optional-sampling theorem for bounded stopping times gives
\begin{align}
    \Expect_{\zeros}[X_{t\wedge\tau}\mid\mathcal F_1]
    \leq X_1.
    \label{eqn:conditional-ineqn}
\end{align}
Consider the event $E_-:=\{r_1\in S,\ A_1\cap S=\emptyset\}$.

\begin{claim}\label{cla:prob-eminus}
    We have $\Pr(E_-) > 0$.
\end{claim}

We first finish the instability proof and then present the proof of \Cref{cla:prob-eminus}.
Note that the event $E_-$ belongs to $\mathcal F_1$, and $X_1=-1$ on $E_-$.
Multiplying the conditional inequality \eqref{eqn:conditional-ineqn} by $\indc{E_-}$ and taking expectations therefore yields
\[
    \Expect_{\zeros}[\indc{E_-}X_{t\wedge\tau}]
    \leq \Expect_{\zeros}[\indc{E_-}X_1]
    =-\Pr(E_-)<0
    \qquad\text{for every }t\geq1.
\]
On the other hand, $\Expect_{\zeros}[\tau]<\infty$ implies $\tau<\infty$ almost surely.
Since $\bds_{\tau} = \zeros$, we have $X_\tau=0$.
Consequently, $X_{t\wedge\tau}\to0$ almost surely as $t\to\infty$.
Moreover, $X_0=0$ and $|X_{k+1}-X_k|\leq1$, so
\[
    |\indc{E_-}X_{t\wedge\tau}|
    \leq |X_{t\wedge\tau}|
    \leq t\wedge\tau
    \leq\tau.
\]
Thus $\tau$ is an integrable bound independent of $t$.
The dominated-convergence theorem gives
\[
    \lim_{t\to\infty}
    \Expect_{\zeros}[\indc{E_-}X_{t\wedge\tau}]
    =0,
\]
contradicting the uniform upper bound $-\Pr(E_-)<0$.
Hence $\Expect_{\zeros}[\tau_{\zeros}^+]=\infty$.
Therefore, no admissible stationary policy stabilizes the system.
For related (super)martingale criteria for non-positive recurrence of countable Markov chains, see~\cite[Chapter~2]{fayolle1995topics}.

\begin{proof}[Proof of \Cref{cla:prob-eminus}]
Decomposing according to whether the requester lies in $S$ and whether $A_1$ intersects $S$, we have
\begin{align*}
    \Pr(A_1\cap S\neq\emptyset)
    &=\Pr(r_1\in S,\ A_1\cap S\neq\emptyset)
      +\Pr(r_1\notin S,\ A_1\cap S\neq\emptyset),\\
    \sum_{i\in S}p_i=\Pr(r_1\in S)
    &=\Pr(r_1\in S,\ A_1\cap S\neq\emptyset)+\Pr(E_-).
\end{align*}
Subtracting the second equality from the first gives
\[
    \Pr(r_1\notin S,\ A_1\cap S\neq\emptyset)-\Pr(E_-)
    =\Pr(A_1\cap S\neq\emptyset)-\sum_{i\in S}p_i\leq0.
\]
If $\Pr(E_-)=0$, then we have $\Pr(r_1\notin S,\ A_1\cap S\neq\emptyset) = 0$.
Since every $p_i>0$ and the $d$ provider draws are independent, this forces $\weight(i,S)=\weight(i,[n])$ for $i\in S$ and $\weight(i,S)=0$ otherwise.
Hence, $S$ is sampling-separating and attains equality, contrary to its choice.
\end{proof}

\section{Deferred proofs for \Cref{sec:sbm-model}}\label{sec:proof-sec-sbm}

We first provide the rate-function properties in \Cref{sec:proof-lem-char-Iit} and the exposure estimates in \Cref{sec:sbm-exposure-estimates}, which will be used by later proofs in this section.

\subsection{Properties of the exposure-rate function}\label{sec:proof-lem-char-Iit}

\begin{lemma}\label{lem:char-Iit}
    For each $i \in [m]$, let $\xi_i := \sum_{j=1}^m \alpha_j c_{ij} / \lambda_j$.
    Then, $I_i(\cdot)$ is continuous and strictly decreasing over $[0, \xi_i]$ with $I_i(0) = \lambda_i$ and $I_i(t) = 0$ for all $t \geq \xi_i$.
    Moreover, if $\lambda_1 = \ldots = \lambda_m = \lambda$, then, for every $i \in [m]$, $\xi_i = 1$ and $I_i(t) = \lambda(1 - t + t \ln t)$ for $t \in [0, 1]$.
\end{lemma}

\begin{proof}

Fix $i$ and write $y_j:=\alpha_jc_{ij}$ and
\[
    F(\bdx):=\sum_j\left(x_j\ln(x_j/y_j)-x_j+y_j\right),
\]
with the convention $0\ln0=0$.
This strictly convex function has unique minimum $F(\bdy)=0$.
Hence $I_i(0)=\sum_jy_j=\lambda_i$ and $I_i(t)=0$ for $t\geq\xi_i$.
For $0\leq t_1<t_2\leq\xi_i$, a minimizer $\bdx'$ at $t_1$ exists by compactness and differs from $\bdy$.
Moving a sufficiently small positive distance from $\bdx'$ toward $\bdy$ gives a feasible point at $t_2$ with strictly smaller objective, proving $I_i(t_2)<I_i(t_1)$.
The same argument shows that the constraint binds for the minimizer at $t<\xi_i$.

The continuity of $I_i(\cdot)$ in the interior of $[0, \xi_i]$ follows from convexity of the value function.
At $t = 0$, the feasible vectors converge uniformly to $\zeros$ as $t\downarrow0$.
At $t = \xi_i$, the feasible vectors $(t/\xi_i)\bdy$ have objective tending to zero.
This proves the continuity of $I_i(\cdot)$ at both endpoints.

When $\lambda_1= \ldots = \lambda_m=\lambda$, we have $\xi_i=1$ since $\sum_{j=1}^m \alpha_j c_{ij} = \lambda_i = \lambda$.
For $0\leq t\leq1$, the log-sum inequality \cite[Theorem 2.7.1]{cover1991elements} gives
\[
    \sum_j x_j\ln(x_j/y_j)\geq t\lambda\ln t
    \qquad\text{if }\sum_jx_j=t\lambda,
\]
with equality attained by $\bdx=t\bdy$.
Thus $I_i(t)=\lambda(1-t+t\ln t)$.
\end{proof}

\subsection{Exposure estimates and singleton obstructions}\label{sec:sbm-exposure-estimates}

Let $\lambda_{\min}:=\min_j\lambda_j$.
For each vertex $u$, define the weighted degree
\[
 T_u:=\frac1{L_n}\sum_{j=1}^m\frac{|N(u,H_j)|}{\lambda_j}.
\]
We first bound the lower tail of $T_u$ and compare $T_u$ with the actual column exposure $B_u$, then establish the singleton obstructions used later in the instability proofs.

The next lemma gives an upper bound on the probability that a fixed vertex in $H_i$ has weighted degree at most $t$, expressed in terms of $I_i(t)$.

\begin{lemma}\label{lem:weighted-degree-tail}
Suppose that $L_n$ satisfies $L_0\leq L_n\leq n / c_{\max}$, where $c_{\max}:=\max_{i,j}c_{ij}$, and $L_0>0$ is fixed.
Then, for every $i\in[m]$, $u\in H_i$, and $0<t<\xi_i$,
\begin{equation}\label{eqn:weighted-degree-tail}
 \Prob(T_u\leq t)\leq \exp\bigl(1-L_nI_i(t)\bigr).
\end{equation}
\end{lemma}
\begin{proof}
Put $\mu_j=\alpha_jc_{ij}$.
The unconstrained minimizer of the optimization for $I_i(\cdot)$ is $\boldsymbol{\mu}:=(\mu_j)_{j\in[m]}$, which satisfies $\sum_j\mu_j/\lambda_j=\xi_i>t$.
Thus the constraint for $I_i(t)$ is active under the unique minimizing profile for $I_i(t)$, and Lagrange multipliers give the profile
\begin{equation}\label{eqn:tilted-degree-profile}
 x_j(t)=\mu_j e^{-\theta/\lambda_j},\qquad
 \sum_j\frac{x_j(t)}{\lambda_j}=t,
 \qquad
 I_i(t)=\sum_j\mu_j(1-e^{-\theta/\lambda_j})-\theta t,
\end{equation}
where $\theta=\theta_i(t)>0$ is the unique Lagrange multiplier determined by the active constraint.
Since $|N(u,H_j)|$ are independent binomial variables, exponential Markov and $\ln(1+z)\leq z$ give
\begin{align*}
 \ln\Prob(T_u\leq t)
 &\leq\theta L_nt+
 \sum_j(n_j-\indc{i=j})\ln\bigl(1+Q_{ij}(e^{-\theta/\lambda_j}-1)\bigr)\\
 &\leq-L_nI_i(t)+Q_{ii}(1-e^{-\theta/\lambda_i})
 \leq-L_nI_i(t)+1,
\end{align*}
as desired.
\end{proof}

We next compare $T_u$ with the actual exposure $B_u$.
The first estimate controls how much $T_u$ can exceed $B_u$ uniformly over all vertices, and the second bounds the expected absolute error conditional on a prescribed neighborhood.

\begin{lemma}\label{lem:exposure-comparisons}
For every fixed $b>0$, if $L_n\geq b\ln n$, then
\begin{equation}\label{eqn:sbm-one-sided-exposure}
 \max_{u\in[n]}(T_u-B_u)_+\toprob0.
\end{equation}
Also, for any fixed constants $0<c_1<c_2<\infty$, as $L_n\to\infty$,
\begin{equation}\label{eqn:sbm-conditional-exposure}
 \Expect\bigl[|B_u-T_u|\mid N(u)=S\bigr]=O(L_n^{-1})
\end{equation}
for every vertex $u$ and every possible neighborhood $S$ such that $c_1L_n\leq|S|\leq c_2L_n$ and $\Pr(N(u)=S)>0$.
\end{lemma}
\begin{proof}
Write $L:=L_n$.
Fix $u\in H_i$ and condition only on a possible neighborhood $N(u)=S$ with $|S|\leq CL$, where $C>0$ is fixed.
All edges not incident to $u$ remain independent.
For all $j \in [m]$ and $v\in S\cap H_j$, we have $\Expect[a_v\mid S]=\lambda_jL+2-Q_{ij}-Q_{jj}$.
The tangent inequality $1/y\geq1/z-(y-z)/z^2$ for $y,z>0$ gives
\[
 B_u\geq T_u-Z_u,\qquad \text{where }
 Z_u:=\sum_{j=1}^m\sum_{v\in S\cap H_j}
       \frac{a_v-\lambda_jL}{\lambda_j^2L^2}.
\]
Conditional on $S$, we have $\Expect Z_u=O(L^{-1})$, and $Z_u-\Expect[Z_u\mid S]$ is a sum of independent centered Bernoulli terms.
Each edge has coefficient at most $2/(\lambda_{\min}^2L^2)$, while the sum of the presence probabilities of the relevant edges is at most $|S|c_{\max}L=O(L^2)$.
Thus $\Var(Z_u\mid S)=O(L^{-2})$.
Bernstein's inequality yields, for each fixed $\epsilon>0$ and all large $n$,
\[
 \Prob(Z_u>\epsilon\mid S)\leq e^{-c_\epsilon L^2}.
\]
If $L\geq b\ln n$, a Chernoff bound \cite[Theorem~2]{janson2016large} and a union bound give $\max_{u \in [n]}|N(u)|\leq CL$ with high probability for a sufficiently large fixed $C$.
Averaging over $S$ and union bounding over all vertices $u$, using $ne^{-c_\epsilon L^2}\to0$, proves \eqref{eqn:sbm-one-sided-exposure}.

To prove \eqref{eqn:sbm-conditional-exposure}, fix $0<c_1<c_2<\infty$ and apply the above conditional mean and variance bounds for $Z_u$ with $C=c_2$.
For $v\in S\cap H_j$, write $a_v=2+X_v$, where $X_v$ is a sum of independent Bernoulli variables with mean $\mu_v=\lambda_jL-Q_{ij}-Q_{jj}>0$ for all large $L$.
Its generating function gives
\begin{align*}
 \Expect[a_v^{-1}\mid S]
 =\int_0^1 z\prod_e(1-p_e+p_ez)\,dz
 \leq\int_0^1 e^{-\mu_v(1-z)}\,dz
 \leq\frac1{\mu_v}
 =\frac1{\lambda_jL}+O(L^{-2}),
\end{align*}
where $p_e$ are the presence probabilities of the remaining edges incident to $v$.
Consequently,
\[
 \Expect[B_u-T_u\mid S]\leq\frac1{|S|+1}+O(L^{-1}).
\]
Since $B_u\geq T_u-Z_u$, taking positive parts gives
\[
 (T_u-B_u)_+\leq(Z_u)_+.
\]
Taking conditional expectations and applying Cauchy--Schwarz, together with the conditional mean and variance bounds for $Z_u$, gives
\[
 \Expect[(T_u-B_u)_+\mid S]
 \leq\Expect[(Z_u)_+\mid S]
 \leq|\Expect[Z_u\mid S]|+\sqrt{\Var(Z_u\mid S)}=O(L^{-1}).
\]
Using $|y|=y+2(-y)_+$ proves \eqref{eqn:sbm-conditional-exposure} for every $S$ with $c_1L\leq|S|\leq c_2L$.
\end{proof}

The following lemma shows that, when $L_n\to\infty$ and $L_nI_i(t)$ is below $\ln n$ by a fixed proportion, a second-moment argument and the conditional comparison in \Cref{lem:exposure-comparisons} yield the low-exposure vertex.

\begin{lemma}\label{lem:low-exposure-witness}
Fix $i$, $0<t<\xi_i$, and $\delta>0$.
Suppose that $L_n\to\infty$ and $L_nI_i(t)\leq(1-\delta)\ln n$.
Then, for every fixed $\epsilon>0$, with high probability, there exists $u\in H_i$ such that $B_u\leq t+\epsilon$.
\end{lemma}
\begin{proof}
Write $L:=L_n$ and fix $\epsilon>0$.
We will construct many vertices with $T_u\leq t$ and show that only a vanishing fraction have $B_u>T_u+\epsilon$.
Using the positive minimizing profile $x_j(t)$ in \eqref{eqn:tilted-degree-profile}, put $k_j:=\lfloor Lx_j(t)\rfloor$ and call a vertex $u\in H_i$ a candidate if $|N(u,H_j)|=k_j$ for every $j \in [m]$.
Since each $x_j(t)>0$ and $\sum_jx_j(t)/\lambda_j=t$, the choice of floors ensures that every candidate satisfies
\[
 T_u=\frac1L\sum_j\frac{k_j}{\lambda_j}\leq t,
 \qquad |N(u)|=\sum_jk_j=\Theta(L).
\]

By \Cref{lem:char-Iit}, $I_i(t)>0$, so the hypothesis implies $L=O(\ln n)$.
For a fixed $u\in H_i$, independence of its community neighbor counts gives
\[
 h:=\Prob(u\text{ is a candidate})
 =\prod_{j=1}^m\Prob\!\left(\operatorname{Bin}(n_j-\indc{i=j},Q_{ij})=k_j\right).
\]
Stirling's formula and the minimizing property of $(x_j(t))$ yield
\[
 h=\exp\{-LI_i(t)+O(\ln L+L^2/n)\}.
\]
Since $\ln L+L^2/n=o(\ln n)$, the number $X$ of candidates satisfies
\[
 \Expect X=n_ih\geq n^{\delta+o(1)}\longrightarrow\infty.
\]

To control fluctuations of $X$, we use the second-moment argument from \cite[Lemma~3.4]{kolokolnikov2014algebraic}.
For distinct $u,v\in H_i$, let $h_0,h_1$ be the probabilities that $u$ is a candidate conditional on $W(u,v)=0$ and $W(u,v)=1$, respectively.
By symmetry, these are also the corresponding probabilities for $v$, and
\[
 h=(1-Q_{ii})h_0+Q_{ii}h_1.
\]
Excluding their shared edge leaves $n_i-2$ possible neighbors in $H_i$, so it holds that
\[
 \frac{h_1}{h_0}
 =\frac{k_i}{n_i-1-k_i}\frac{1-Q_{ii}}{Q_{ii}}
 \longrightarrow\frac{x_i(t)}{\alpha_ic_{ii}}
 =e^{-\theta_i(t)/\lambda_i}<1.
\]
Thus $h_1\leq h_0$ for all sufficiently large $n$.
Conditional on the status of edge $\{u, v\}$, the remaining incident edges of $u$ and $v$ are independent, giving
\begin{align*}
 \Prob(u,v\text{ are candidates})
 =(1-Q_{ii})h_0^2+Q_{ii}h_1^2
 \leq h_0h\leq\frac{h^2}{1-Q_{ii}}.
\end{align*}
Consequently,
\[
 \frac{\Var X}{(\Expect X)^2}
 \leq\frac1{\Expect X}+\frac{Q_{ii}}{1-Q_{ii}}\longrightarrow0,
\]
because $\Expect X\to\infty$ and $Q_{ii}=c_{ii}L/n\to0$.
Chebyshev's inequality therefore gives $X/\Expect X\toprob1$.

It remains to compare $B_u$ with $T_u$ at the candidate vertices.
Let $Y$ be the number of candidates with $B_u>T_u+\epsilon$.
Every possible neighborhood $S$ of a candidate has size $\Theta(L)$, so \eqref{eqn:sbm-conditional-exposure} and Markov's inequality give, uniformly over these $S$,
\[
 \Prob(B_u>T_u+\epsilon\mid N(u)=S)
 \leq\frac{\Expect[|B_u-T_u|\mid N(u)=S]}{\epsilon}
 =O(L^{-1}).
\]
Averaging over these neighborhoods and summing over $u\in H_i$ yields $ \Expect Y=O(L^{-1})\Expect X=o(\Expect X). $
Concentration of $X$ and Markov's inequality for $Y$ imply that, with high probability, $X\geq\Expect X/2$ and $Y<\Expect X/2$.
Hence $X-Y>0$ with high probability, and any candidate not counted by $Y$ satisfies $B_u\leq T_u+\epsilon\leq t+\epsilon$.
\end{proof}

The preceding lemma treats $L_n\to\infty$.
For bounded $L_n$, the next lemma obtains a deficient singleton cut from an isolated star component of fixed size.

\begin{lemma}\label{lem:bounded-degree-obstruction}
For every fixed $M\geq L_0>0$, suppose that $L_n$ satisfies $L_0\leq L_n\leq M$.
Then, with high probability, a singleton cut violates \eqref{eq:er-closed-neighborhood-cut-SBM}.
\end{lemma}
\begin{proof}
Fix a community $H_i$ and put $k=d2^d$.
There are $\Theta(n^{k+1})$ star specifications inside $H_i$.
Each requires $k$ present edges and $O(n)$ absent edges, so its probability of forming an isolated component is $\Theta(n^{-k})$, uniformly on the stated interval.
Thus the number $X$ of such components has $\Expect X=\Theta(n)$.
Distinct overlapping component specifications cannot coexist, and disjoint ones have joint-to-product probability ratio $(1-Q_{ii})^{-(k+1)^2}=1+O(n^{-1})$ because their mutual absent edges are counted twice in the product.
Therefore $\Var X/(\Expect X)^2\leq1/\Expect X+O(n^{-1})=o(1)$.
For any leaf $u$ of such a star,
\[
 \sum_v h_d\left(\frac{W(v,u)}{a_v}\right)
 =h_d(1/2)+h_d(1/(k+1))
 \leq1-2^{-d}+\frac{d}{k+1}<1,
\]
so its singleton cut is deficient.
\end{proof}

\subsection{Proof of \Cref{lem:block-cut-instability}}\label{sec:proof-lem-block-cut-instability}

First suppose $L_n\to\infty$.
Choose a block indicator $\bdx$ satisfying \eqref{eqn:bc-general-violate}.
Set $S=\bigcup_{i:x_i=1}H_i$.
For $u\in H_i$, concentration of the finitely many block-neighbor counts gives $W(u,S)/a_u\toprob(R\bdx)_i$.
These ratios are bounded by one, so the convergence also holds in mean, uniformly over vertices of each fixed type.
Since $h_d$ is $d$-Lipschitz, averaging the absolute errors and applying Markov's inequality yields
\[
 \frac1n\sum_u h_d\left(\frac{W(u,S)}{a_u}\right)
 \toprob\sum_i\alpha_i h_d((R\bdx)_i)
 <\sum_i\alpha_ix_i=\frac{|S|}{n}.
\]
Thus a cut violates the necessary inequality with high probability.
For bounded $L_n$, \Cref{lem:bounded-degree-obstruction} gives a deficient singleton cut.
Every subsequence has a further subsequence on which $L_n$ is bounded or tends to infinity, so the two cases cover all admissible scales.
\Cref{thm:min-token-backpressure} proves instability.

\subsection{\texorpdfstring{Proof of \Cref{thm:instability-sbm}}{Proof of the exposure-threshold instability proposition}}\label{sec:proof-thm-instability-sbm}

Define $J_* := \min_{i \in [m]} I_i(1 / d)$, and assume $\limsup_{n\to\infty}L_nJ_*/\ln n<1$.
If $J_*=0$, nonnegativity of the rate functions implies $I_j(1/d)=0$ for some $j$.
By \Cref{lem:char-Iit}, this equality holds only when $1/d\geq\xi_j$, so $\xi_j\leq1/d$.
Since $\xi_1=1>1/d$ when $m=1$, this case has $m\geq2$.
Using $h_d(z)<dz$ for $0<z\leq1$ and symmetry of $(c_{ij})$ gives
\[
 \sum_i\alpha_i h_d(R(i,j))<d\sum_i\alpha_iR(i,j)
 =d\alpha_j\xi_j\leq\alpha_j.
\]
This strictly violates \eqref{eqn:bc-general-sbm}, so \Cref{lem:block-cut-instability} implies instability.

For the remaining case, suppose $J_*>0$.
First suppose $L_n\to\infty$, and choose a community $i$ attaining $J_*$.
Because $I_i(1/d)=J_*>0$ and $\limsup_{n\to\infty}L_nJ_*/\ln n<1$, continuity of $I_i$ allows fixed $t\in(0,1/d)$ and $\delta>0$ such that $L_nI_i(t)\leq(1-\delta)\ln n$ eventually.
Taking $0<\epsilon<1/d-t$ in \Cref{lem:low-exposure-witness} gives a vertex with $B_u<1/d$ with high probability.
Since $h_d(x)\leq dx$ for $x\in[0,1]$,
\[
 \sum_v h_d\left(\frac{W(v,u)}{a_v}\right)
 \leq d\sum_v\frac{W(v,u)}{a_v}=dB_u<1.
\]
Hence the singleton cut $\{u\}$ violates \eqref{eq:er-closed-neighborhood-cut-SBM}.
For bounded $L_n$, we combine \Cref{thm:min-token-backpressure} with \Cref{lem:bounded-degree-obstruction} to rule out stability.
Every subsequence has a further subsequence on which $L_n$ is bounded or tends to infinity, so these cases cover all admissible sequences.

\subsection{\texorpdfstring{Proof of \Cref{thm:stability-sbm}}{Proof of the SBM stability theorem}}\label{sec:proof-thm-stability-sbm}\label{sec:sbm-cut-estimates}

We first develop the block-approximation result.
Define the normalization matrix $D:=\diag{a_1,\ldots,a_n}$.
We compare the normalized sampling matrix $D^{-1}\weight$ with a deterministic block-averaged proxy $P_R\in[0,1]^{n\times n}$ defined as follows.
For $u \in H_i$ and $v \in H_j$, let
\begin{align*}
    P_R(u, v) := \frac{R(i, j)}{n_j} = \frac{c_{ij}}{n\lambda_i},
\end{align*}
and denote $P_R(u, S) := \sum_{v \in S} P_R(u, v)$ for every $S \subseteq [n]$.
Set
\[
    \rho:=\norm{D^{-1}\weight-P_R}_2,\qquad
    w_{\max}:=\max_{i,j\in[m]}\frac{c_{ij}}{\lambda_i}.
\]
The following lemma controls the block-approximation error.

\begin{lemma}\label{lem:sbm-mixing}
 Suppose $L_n\geq b\ln n$ for some fixed $b>0$, and $\min_u a_u\geq cL_n$ with high probability for some fixed $c>0$.
 Then $\norm{D^{-1}\weight-P_R}_2=O_{\Prob}(L_n^{-1/2})\toprob0$.
\end{lemma}

\begin{proof}
Let $A^\circ$ be the adjacency matrix of the sampled graph, so $\weight=A^\circ+I$.
Let $D_{\bdlambda}$ be diagonal with $D_{\bdlambda}(u,u)=\lambda_iL_n$ for $u\in H_i$.
Set $M:=\weight-D_{\bdlambda}P_R$.
Since $P_R\ones=\ones$ and $\weight\ones=D\ones$,
\[
    D-D_{\bdlambda}=\diag{M\ones}.
\]
Each entry of $P_R$ is at most $w_{\max}/n$, so, writing $\norm{\cdot}_{\rm F}$ for the Frobenius norm,
\begin{align*}
    \norm{(D-D_{\bdlambda})P_R}_2
    &\leq\norm{\diag{M\ones}P_R}_{\rm F}
    \leq\frac{w_{\max}}{\sqrt n}\norm{M\ones}_2
    \leq w_{\max}\norm{M}_2.
\end{align*}
Moreover, $\Expect \weight-D_{\bdlambda}P_R$ is diagonal, with entry $1-Q_{ii}$ at each $u\in H_i$, and hence has operator norm at most one.
Therefore
\begin{align*}
    \norm{D^{-1}\weight-P_R}_2
    &\leq \frac{\norm{W - D P_R}_2}{\min_{u \in [n]} a_u}
    \leq\frac{\norm{M}_2+\norm{(D_{\bdlambda}-D)P_R}_2}{\min_{u \in [n]} a_u}
    \leq \frac{1+w_{\max}}{\min_{u \in [n]} a_u} \cdot \norm{M}_2 \\
    &\leq \frac{1+w_{\max}}{\min_{u \in [n]} a_u}\bigl(\norm{W - \Expect W}_2 + \norm{\Expect W - D_{\bdlambda} P_R}_2\bigr)
    \leq\frac{1+w_{\max}}{\min_{u \in [n]} a_u}\bigl(\norm{A^\circ-\Expect A^\circ}_2+1\bigr).
\end{align*}
By \cite[Theorem~5.2]{lei2015consistency}, applied with $\Delta_n=c_{\max}L_n$, we have $\norm{A^\circ-\Expect A^\circ}_2=O_{\Prob}(\sqrt{L_n})$.
Its hypotheses hold because the undirected edges are independent, $n\max_{u\ne v}\Pr(\{u,v\}\in E)\leq\Delta_n$, and $\Delta_n\geq c_{\max}b\ln n$.
\footnote{See also~\cite{feige2005spectral} for the analogous bound for Erd\H{o}s--R\'enyi graphs.}
The assumed lower bound $\min_u a_u\geq cL_n$ with high probability therefore yields $ \norm{D^{-1}\weight-P_R}_2=O_{\Prob}(L_n^{-1/2})\toprob0, $ as desired.
\end{proof}

Next, we record two deterministic estimates that hold simultaneously for every $S\subseteq[n]$.
The operator norm gives $\norm{(D^{-1}\weight-P_R)\ones_S}_2\leq\rho\sqrt{|S|}$, while each entry of $P_R$ is at most $w_{\max}/n$, so $   \norm{P_R\ones_S}_2^2\leq\frac{w_{\max}^2|S|^2}{n}. $
Consequently,
\begin{align}
    \norm{D^{-1}\weight\ones_S}_2^2
    \leq 2\norm{(D^{-1}\weight-P_R)\ones_S}_2^2+2\norm{P_R\ones_S}_2^2
    \leq 2\left(\rho^2+w_{\max}^2\frac{|S|}{n}\right)|S|.
    \label{eqn:sbm-sum-hdpu-case1-quad}
\end{align}
Also, the $d$-Lipschitz property of $h_d$ and Cauchy--Schwarz give
\begin{align}
    \left|\sum_{u=1}^n h_d\left(\frac{\weight(u,S)}{a_u}\right)-\sum_{u=1}^n h_d(P_R(u,S))\right|
    &\leq d \sum_{u=1}^n \left| \frac{W(u, S)}{a_u} - P_R(u, S) \right| \notag \\
    &\leq d\sqrt n\,\norm{(D^{-1}\weight-P_R)\ones_S}_2
    \leq d\rho\sqrt{n|S|}\leq d\rho n.
    \label{cor:det-con-sbm-mixing}
\end{align}

To pass from whole communities to arbitrary cuts, we extend the block condition to fractional community indicators.

\begin{lemma}\label{lem:equi-con-bc-general-sbm}
    If \eqref{eqn:bc-general-sbm} holds, then
    \begin{align*}
        \sum_{i=1}^m \alpha_i h_d((R\bdx)_i)
        > \sum_{i=1}^m \alpha_i x_i, \quad \forall \bdx \in [0, 1]^m \setminus \{\zeros, \ones\}.
    \end{align*}
\end{lemma}

\begin{proof}
Define
\begin{align}
    F(\bdx):=\sum_{i=1}^m\alpha_i h_d((R\bdx)_i)-\sum_{i=1}^m \alpha_i x_i.
    \label{eqn:def-fx-diff-cut}
\end{align}
The function $F$ is concave with $F(\zeros)=F(\ones)=0$.
In addition, by \eqref{eqn:bc-general-sbm}, $F(\bdx) > 0$ for every $\bdx \in \{0, 1\}^m \setminus \{\zeros, \ones\}$.
Since every point of $[0,1]^m$ can be written as a convex combination of points in $\{0, 1\}^m$, we know that $F$ is nonnegative throughout $[0, 1]^m$.
Fix $\bdx \in [0, 1]^m \setminus \{0, 1\}^m$, and we show that $F(\bdx) > 0$.
Choose $j \in [m]$ such that $x_j\in(0,1)$, and let $\bdx^{(0)}$ and $\bdx^{(1)}$ be obtained by replacing the $j$-th coordinate of $\bdx$ by zero and one, respectively.
Since $d\geq2$ and every entry of $R$ is strictly positive, $F$ is strictly concave along each coordinate line.
Hence $ F(\bdx)>(1-x_j)F(\bdx^{(0)})+x_jF(\bdx^{(1)})\geq0, $ concluding the proof.
\end{proof}

Recall that $\rho=\| D^{-1}\weight-P_R \|_2$ measures the error of the block approximation.
The following lemma combines a small $\rho$, a positive exposure margin above $1/d$, and \eqref{eqn:bc-general-sbm} to deterministically verify \eqref{eq:er-closed-neighborhood-cut-SBM} for every nonempty proper cut.

\begin{lemma}\label{lem:deterministic-sbm-cuts}
 Suppose \eqref{eqn:bc-general-sbm} holds.
 For every $\eta>0$ there is $\epsilon>0$ such that $\min_uB_u\geq1/d+\eta$ and $\rho\leq\epsilon$ imply \eqref{eq:er-closed-neighborhood-cut-SBM} for every nonempty proper cut.
\end{lemma}

\begin{proof}
First choose sufficiently small $\theta_0\in(0,1/2)$, and then sufficiently small $\epsilon>0$ such that
\begin{align}
    2\binom d2(\epsilon^2+w_{\max}^2\theta_0)<d\eta,
    \qquad 2(\epsilon^2+w_{\max}^2\theta_0)<1.
    \label{eqn:con-choose-thm-3}
\end{align}
Suppose $\rho\leq\epsilon$ and $\min_{u\in[n]}B_u\geq1/d+\eta$.
Then, for every nonempty $S\subseteq[n]$,
\begin{align}
    \sum_{u=1}^n\frac{\weight(u,S)}{a_u}=\sum_{v\in S}B_v
    \geq |S| \cdot \min_{u \in S} B_u
    \geq\left(\frac1d+\eta\right)|S|.
    \label{eq:q-sum-lower-SBM}
\end{align}
We verify \eqref{eq:er-closed-neighborhood-cut-SBM} for every nonempty proper cut $S\subseteq[n]$.

\paragraph{Case 1: $0<|S|/n<\theta_0$.}
Since $h_d(x)\geq dx-\binom d2x^2$ for $x\in[0,1]$, \eqref{eqn:sbm-sum-hdpu-case1-quad} and \eqref{eq:q-sum-lower-SBM} give
\begin{align*}
    \sum_{u=1}^n h_d\left(\frac{\weight(u,S)}{a_u}\right)
    &\geq d\sum_{u=1}^n \frac{W(u, S)}{a_u}-\binom d2\norm{D^{-1}\weight\ones_S}_2^2\\
    &\geq\left[1+d\eta-2\binom d2\left(\rho^2+w_{\max}^2\frac{|S|}{n}\right)\right]|S|
    >|S|,
\end{align*}
where the last inequality holds by the first condition in \eqref{eqn:con-choose-thm-3}.

\paragraph{Case 2: $\theta_0\leq|S|/n\leq1-\theta_0$.}
Consider the function $F$ defined in \eqref{eqn:def-fx-diff-cut}.
Since the set
\[
    \Gamma:=\{\bdx\in[0,1]^m:\boldsymbol\alpha^{\mathsf T}\bdx\in[\theta_0,1-\theta_0]\}
\]
is compact, and $F(\bdx) > 0$ for every $\bdx \in \Gamma$ by \Cref{lem:equi-con-bc-general-sbm}, we can define $\delta:=\min_{\bdx\in\Gamma}F(\bdx)>0$.
Set $x_i^S:=|S\cap H_i|/n_i$, so $\bdx^S\in\Gamma$ and $P_R(u,S)=(R\bdx^S)_i$ for $u\in H_i$.
It holds that
\begin{align*}
    \sum_{u=1}^n h_d\left(\frac{\weight(u,S)}{a_u}\right)-|S|
    &\geq \sum_{u=1}^n h_d(P_R(u,S)) - |S| - d \rho n  \\
    &\geq nF(\bdx^S)-d\rho n
    \geq n(\delta-d\rho)>0,
\end{align*}
where the first inequality holds by \eqref{cor:det-con-sbm-mixing}, the second inequality holds since $\sum_{i=1}^m \alpha_i x_i^S = |S| / n$, and the last inequality holds by choosing $\epsilon$ small enough so that $\rho \leq \epsilon<\delta/d$.

\paragraph{Case 3: $1-\theta_0<|S|/n<1$.}
Let $T:=[n]\setminus S$, so $0<|T|<\theta_0n$.
Since $d\geq2$, it holds that
\[
    \sum_{u=1}^n\left(\frac{\weight(u,T)}{a_u}\right)^d
    \leq\norm{D^{-1}\weight\ones_T}_2^2
    \leq2\left(\rho^2+w_{\max}^2\frac{|T|}{n}\right)|T|<|T|,
\]
where the second inequality holds by \eqref{eqn:sbm-sum-hdpu-case1-quad}.
As $h_d(\weight(u,S)/a_u)=1-(\weight(u,T)/a_u)^d$, this is precisely the desired cut inequality for $S$.
\end{proof}

Now, we are ready to prove \Cref{thm:stability-sbm}.

\begin{proof}[Proof of \Cref{thm:stability-sbm}]
Define $J_* := \min_{i \in [m]} I_i(1 / d)$.
Assume \eqref{eqn:bc-general-sbm} holds and $\liminf_{n\to\infty}L_nJ_*/\ln n>1$.
For $m\geq2$, applying \eqref{eqn:bc-general-sbm} to $\bdx=\boldsymbol e_j$ gives
\[
 \alpha_j<\sum_{i=1}^m\alpha_i h_d(R(i,j))
 <d\sum_{i=1}^m\alpha_iR(i,j)=d\alpha_j\xi_j,
\]
where the second inequality holds since $h_d(z)<dz$ for $0<z\leq1$, and the equality holds by the symmetry of $(c_{ij})$.
For $m=1$, $\xi_1=1>1/d$ directly.
Thus $\xi_i>1/d$ for every $i \in [m]$.
By \Cref{lem:char-Iit}, $I_i$ is strictly decreasing on $[0,\xi_i]$ with $I_i(\xi_i)=0$, so $I_i(1/d)>0$ for every $i$.
Taking the minimum over the finitely many communities gives $J_*>0$.
Choose $b>0$ with $1<bJ_*<\liminf_{n\to\infty}L_nJ_*/\ln n$, and then choose $t\in(1/d,\min_i\xi_i)$ so that $b\min_iI_i(t)>1$.
Then $L_n\geq b\ln n$ eventually, and \eqref{eqn:weighted-degree-tail} and a union bound imply $\min_uT_u>t$ with high probability.
By \eqref{eqn:sbm-one-sided-exposure}, $\min_uB_u\geq1/d+\eta$ with high probability for some fixed $\eta>0$.
Also,
\[
 a_u\geq|N(u)|\geq\lambda_{\min}L_nT_u
 \geq\lambda_{\min}L_nt
\]
on the same event.
\Cref{lem:sbm-mixing} gives $\rho\toprob0$, so \Cref{lem:deterministic-sbm-cuts} verifies every strict cut inequality.
Stability follows from \Cref{thm:min-token-backpressure}.
\end{proof}

\subsection{Proof of \Cref{thm:special-equal-degree-coeffi}}\label{sec:proof-cor-special-equal-degree-coeffi}

The Erd\H{o}s--R\'enyi graph $G(n,\lambda\ln n/n)$ is the single-community case of the stochastic block model, with $m=1$, $\alpha_1=1$, and $c_{11}=\lambda$.
Thus $\lambda_1=\lambda$ and $R=(1)$, and \eqref{eqn:bc-general-sbm} is vacuous because there are no nontrivial block indicators.
By \Cref{lem:char-Iit}, the exposure rate function admits the simple form $I_1(t)=\lambda(1-t+t\ln t)$ for $t \in [0, 1]$.
As a result, $ I_1(1/d)>1$ if and only if $\lambda>\lambda_d^*, $ with the same equivalence for the reverse inequalities.
\Cref{thm:instability-sbm} and \Cref{thm:stability-sbm} with $L_n=\ln n$ then conclude the instability and stability claims, respectively.

\bibliographystyle{alpha}
\bibliography{references}

\end{document}